\documentclass[11pt,a4paper]{article}

\usepackage{amssymb}
\usepackage{amsmath}
\usepackage{epsfig}
\usepackage{graphicx}
\usepackage{color}
\usepackage{bigints}
\usepackage{array}

\graphicspath{{./}{figures/}}

\newcommand{\eps}{\varepsilon}
\newcommand{\pr}[1]{\mathrm{Pr}\left[#1\right]}
\newcommand{\cancel}[1]{}

\newenvironment{emromani}
{
	
	\begin{enumerate}}
	{\end{enumerate}}

\newtheorem{theorem}{Theorem}[section]
\newtheorem{lemma}{Lemma}[section]

\newbox\ProofSym
\setbox\ProofSym=\hbox{%
	\unitlength=0.18ex%
	\begin{picture}(10,10)
	\put(0,0){\framebox(9,9){}}
	\put(0,3){\framebox(6,6){}}
	\end{picture}}

\begin{document}

\title{Extensions of Self-Improving Sorters\thanks{A preliminary version appeared in Proceedings of the International Symposium on Algorithms and Computation, 2018~\cite{cheng}.}}
\author{Siu-Wing Cheng\thanks{Department of Computer Science and Engineering, HKUST, Hong Kong.  Supported by Research Grants Council, Hong Kong, China (project no.~16200317)} \quad \quad Kai Jin\footnotemark[2] \quad\quad Lie Yan\thanks{Hangzhou, China.  Part of the work was conducted while the author was at HKUST and supported by the Hong Kong PhD Fellowship.}}

\date{}
\maketitle

\begin{abstract}
Ailon et~al.~(SICOMP~2011) proposed a self-improving sorter that tunes its performance to an unknown input distribution in a training phase.  The input numbers $x_1,x_2,\ldots,x_n$ come from a product distribution, that is, each $x_i$ is drawn independently from an arbitrary distribution ${\cal D}_i$.  We study two relaxations of this requirement.  The first extension models hidden classes in the input.  We consider the case that numbers in the same class are governed by linear functions of the same hidden random parameter.  The second extension considers a hidden mixture of product distributions.  
\end{abstract}

\section{Introduction}

Self-improving algorithms proposed by Ailon et al.~\cite{ailon11} can tune their computational performance to the input distribution.  There is a \emph{training phase} in which the algorithm learns certain input features and computes some auxiliary structures.  After the training phase, the algorithm uses these auxiliary structures in the \emph{operation phase} to obtain an expected time complexity that is no worse and possibly smaller than the best worst-case complexity known.  The expected time complexity in the operation phase is called the \emph{limiting complexity}.

This computational model addresses two issues.  First, the worst-case scenario may not happen, so the best time complexity for the input encountered may be smaller than the worst-case optimal bound.  Second, previous efforts for mitigating the worst-case scenarios often consider average-case complexities, and the input distributions are assumed to be simple distributions like Gaussian, uniform, Poisson, etc.~whose parameters are given beforehand.  In contrast, Ailon et al.~only assume that individual input items are independently distributed, while the distribution of an input item can be arbitrary.  No other information is needed.

The problems of sorting and two-dimensional Delaunay triangulation are studied by Ailon et al.~\cite{ailon11}.  An input instance $I$ for the sorting problem has $n$ numbers.  The $i$-th number $x_i$ is drawn independently from a hidden distribution ${\cal D}_i$.  The joint distribution $\prod_{i=1}^n {\cal D}_i$ is called a \emph{product distribution}.  Let $\pi(I)$ denote the sequence of the ranks of the $x_i$'s, which is a permutation of $[n]$.  It is shown that for any $\eps \in (0,1)$, there is a self-improving algorithm with limiting complexity $O(\eps^{-1}(n  + H_\pi))$, where $H_\pi$ is the entropy of the distribution of $\pi(I)$.  By Shannon's theory~\cite{cover06}, any comparison-based sorting algorithm requires $\Omega(n + H_\pi)$ expected time.  The self-improving sorter uses $O(n^{1+\eps})$ space.  The training phase processes $O(n^{\eps})$ input instances in $O(n^{1+\eps})$ time, and it succeeds with probability at least $1 - 1/n$, i.e., the probability of achieving the desired limiting complexity is at least $1 - 1/n$.  For two-dimensional Delaunay triangulations, Ailon et al.~also obtained an optimal limiting complexity for product distributions.
%(within the class of comparison-based algorithms) for an input of $n$ points such that the $i$-th point is drawn from an arbitrary hidden distribution, and the $n$ point distributions are independent from each other.  

Subsequently, Clarkson et al.~\cite{clarkson14} developed self-improving algorithms for two-dimensional coordinatewise maxima and convex hulls, assuming that the input comes from a product distribution.  The limiting complexities for the maxima and the convex hull problems are $O(\mathrm{OptM} + n)$ and $O(\mathrm{OptC} + n\log\log n)$, where OptM and OptC are the expected depths of optimal linear decision trees for the maxima and convex hull problems, respectively.

On one hand, the product distribution requirement is very strong; on the other hand, Ailon et al.~showed that $\Omega(2^{n\log n})$ bits of storage are necessary for optimal sorting if the $n$ numbers are drawn from an arbitrary distribution.  We study two extensions of the input model that are natural and yet possess enough structure for efficient self-improving algorithms to be designed.

The first extension models the situation in which some input elements depend on each other.  We consider a hidden partition of the input $I = (x_1,\cdots,x_n)$ into classes $S_k$'s.  The input numbers in a class $S_k$ are distinct linear functions of the same hidden random parameter $z_k$.  The distributions of the $z_k$'s are arbitrary and each $z_k$ is drawn independently.\footnote{There is a technical condition required of the input distribution to be explained in Section~\ref{sec:linear}.}  We call this model a \emph{product distribution with hidden linear classes}.  Our first result is a self-improving sorter with optimal limiting complexity under this model.

\begin{theorem}
	\label{thm:1}
	For any $\eps \in (0,1)$, there exists a self-improving sorter for any product distribution with hidden linear classes that has a limiting complexity of $O\left(n/\eps + H_\pi/\eps\right)$.   The storage needed by the operation phase is $O(n^2)$.  The training phase processes $O(n^{\eps})$ input instances in $O(n^2\log^3 n)$ time and $O(n^2)$ space.  The success probability is at least $1-1/n$.
\end{theorem}

%Choose any $\eps \in (0,1)$.  Our self-improving sorter has an $O(n/\eps + H_\pi/\eps)$ limiting complexity, uses $O(n^2)$ space, and requires a training phase that processes $O(n^\eps)$ input instances in $O(n^2\log^3 n)$ time with a success probability at least $1-1/n$.

In the second extension, the distribution of $I$ is a mixture $\sum_{q=1}^\kappa \lambda_q {\cal D}_q$, where $\kappa$ and the $\lambda_q$'s are hidden, and every ${\cal D}_q$ is a hidden product distribution of $n$ real numbers.  In other words, over a large collection of input instances, for all $q \in [1,\kappa]$, a fraction $\lambda_q$ of them are expected to be drawn from ${\cal D}_q$.  Although $\kappa$ is unknown, we are given an upper bound $m$ of $\kappa$.  We call this model \emph{a hidden mixture of product distributions}.  Our second result is a self-improving sorter under this model.  
\begin{theorem}
	\label{thm:2}
	For any $\eps \in (0,1)$, there is a self-improving sorter for any hidden mixture of at most $m$ product distributions that has a limiting complexity of $O\left((n\log m)/\eps + H_\pi/\eps \right)$.  The storage needed by the operation phase is $O(mn + m^\eps n^{1+\eps})$. The training phase processes $O(mn\log(mn))$ input instances in $O(mn\log^2(mn) + m^\eps n^{1+\eps}\log(mn))$ time using $O(mn\log(mn) + m^\eps n^{1+\eps})$ space.  The success probability is at least $1-1/(mn)$.
\end{theorem}
In the interesting special case of $m = O(1)$, the limiting complexity is $O(n/\eps + H_\pi/\eps)$ which is optimal.

\section{Hidden linear classes}
\label{sec:linear}

There is a hidden partition of $[n]$ into classes.  For every $i \in [1,n]$, the distribution of $x_i$ is degenerate if $x_i$ is equal to a fixed value.  Each such $x_i$ will be recognized in the training phase.  For the remaining $i$'s, the distributions of $x_i$'s are non-degenerate, and we use $S_1,\cdots, S_g$ to denote the hidden classes formed by them.  Numbers in the same class $S_k$ are generated by linear functions of the same hidden random parameter $z_k$.  Different classes are governed by different random parameters.  We know that the functions are linear, but no other information is given to us.

Let ${\cal D}_k$ denote the distribution of $z_k$.  There is a technical condition that is required of the ${\cal D}_k$'s: there exists a constant $\rho \in (0,1)$ such that for every $k \in [1,g]$ and every $c \in \mathbb{R}$, $\pr{z_k = c} \leq 1-\rho$.  This condition says that ${\cal D}_k$ does not concentrate too much on any single value, which is quite a natural phenomenon.  Our algorithm does not need to know $\rho$, but $\rho$ affects the probabilistic guarantees on the correctness and limiting complexity.  The input size must be at least $e^{3/\rho^2}$ for Theorem~\ref{thm:1} to hold.

\subsection{Training phase}
\label{sec:train}

\subsubsection{Learning the linear classes}

We learn the classes and the linear functions using $3\ln^2 n$ input instances.  Denote these instances by $I_1,I_2,\cdots,I_{3\ln^2 n}$.  Let $x_i^{(a)}$ denote the $i$-th input number in $I_a$.  We first recognize the degenerate distributions by checking which $x_i^{(a)}$ is fixed for $a \in [1,3\ln^2 n]$.   
%Chernoff bound implies that this method works with high probability.

\begin{lemma}
	\label{lem:degenerate}
	Assume that $n \geq e^{2/(3\rho)}$.  It holds with probability at least $1-1/n$ that for all $i \in [1,n]$, if $x^{(a)}_i$ is the same for all $a \in [1,3\ln^2 n]$, the distribution of $x^{(a)}_i$ is degenerate.
\end{lemma}
\begin{proof}
	Let $c_i$ be the observed value of $x^{(a)}_i$ for $a \in [1,3\ln^2 n]$.  If the distribution of $x^{(a)}_i$ is not degenerate, the probability of $x^{(a)}_i = c_i$ for all $a \in [1,3\ln^2 n]$ is at most $(1-\rho)^{3\ln^2 n} \leq e^{-3\rho\ln^2 n} \leq e^{-2\ln n} = n^{-2}$.  Applying the union bound establishes the lemma.	
\end{proof}

Assume that the degenerate distributions are taken out of consideration.  If $i$ and $j$ belong to the same class $S_k$, then $x_i^{(a)}$ and $x_j^{(a)}$ are linearly related as $a$ varies.  Conversely, if $i$ and $j$ belong to different classes, it is highly unlikely that $x_i^{(a)}$ and $x_j^{(a)}$ remain linearly related as $a$ varies because they are governed by independent random parameters.  We check if the triples of points $(x_i^{(a-2)},x_j^{(a-2)})$, $(x_i^{(a-1)},x_j^{(a-1)})$, and $(x_i^{(a)},x_j^{(a)})$ are collinear for every $a \in [3,3\ln^2 n]$ and every distinct pair of $i$ and $j$ from $[1,n]$.  We quantify this intuition in the following result.

\begin{lemma}
	\label{lem:collinear}
	Let $i$ and $j$ be two distinct indices in $[1,n]$ that belong to different classes.  For every $a \in [3,3\ln^2 n]$, let $E_{ij}^{(a)}$ denote the event that the points $(x_i^{(a-2)},x_j^{(a-2)})$, $(x_i^{(a-1)},x_j^{(a-1)})$, and $(x_i^{(a)},x_j^{(a)})$ are not collinear.  For any $n \geq e^{3/\rho^2}$,
	$
	\pr{\bigcup_{a=3}^{3\ln^2 n} E_{ij}^{(a)}} \geq 1-n^{-3}.
	$
	%For $i \in [1,2]$, let $\eps_i \in (0,1)$ be a parameter, and let $\Phi_i$ be a probability distribution on a discrete set of real numbers such that for any $x \in \mathbb{R}$, $\pr{X = x} \leq 1-\eps_i$, where $X$ is a random variable of $\Phi_i$.  Let $X_1$, $X_2$, and $X_3$ be three independent random variables of $\Phi_1$, and let $Y_1$, $Y_2$, and $Y_3$ be three independent random variables of $\Phi_2$.  Then, the probability that the points $(X_1,Y_1)$, $(X_2,Y_2)$, and $X_3,Y_3)$ are not collinear is at least $\eps_1\eps_2$.
\end{lemma}
\begin{proof}
	First, we bound $\pr{E_{ij}^{(3a)}}$ from below for $a \in [1,\ln^2 n]$.  It is well known~\cite[Sections 1.3.3 and 1.5.3]{orourke} that the points $(x_i^{(3a-2)},x_j^{(3a-2)})$, $(x_i^{(3a-1)},x_j^{(3a-1)})$, and $(x_i^{(3a)},x_j^{(3a)})$ are collinear if and only if 
	\begin{equation}
	\left|\begin{array}{ccc}
	x_i^{(3a-2)} & x_j^{(3a-2)} & 1 \\
	x_i^{(3a-1)} & x_j^{(3a-1)} & 1 \\
	x_i^{(3a)} & x_j^{(3a)} & 1
	\end{array}\right| = 0.  \label{eq:0}
	\end{equation}
	Assume that $x_i^{(3a-2)} = c_1$ and $x_i^{(3a-1)} = c_2$ for two fixed values $c_1$ and $c_2$.  Since $i$ and $j$ are in different classes, $x_i^{(b)}$ and $x_j^{(b')}$ are independent for all $b$ and $b'$.  Also, $x_j$ in one instance $I_b$ does not influence $x_j$ in a different instance $I_{b'}$.  So there is no dependence among $x_i^{(3a)}$, $x_j^{(3a-2)}$, $x_j^{(3a-1)}$, and $x_j^{(3a)}$.
	
	Suppose that $c_1 \not= c_2$.  If $E_{ij}^{(3a)}$ does not occur, then by \eqref{eq:0}, we can express $x_j^{(3a)}$ as a function $f(c_1,c_2,x_i^{(3a)},x_j^{(3a-2)},x_j^{(3a-1)})$.  Hence,
	\begin{eqnarray*}
		& & \pr{E_{ij}^{(3a)} | x_i^{(3a-2)} = c_1 \,\wedge\, x_i^{(3a-1)} = c_2 \,\wedge\, c_1 \not= c_2} \\
		& = & \pr{x_j^{(3a)} \not= f(c_1,c_2,x_i^{(3a)},x_j^{(3a-2)},x_j^{(3a-1)})} \\
		& \geq & \rho.
		%& = & 1 - \iiint \pr{x_j^{(3a)} = f(c_1,c_2,u,v,w) | x_i^{(3a)} = u \wedge x_j^{3a-2)} = v \wedge x_j^{(3a-1)} = w} \cdot \\
		%& & \quad\quad\quad\quad\, \pr{x_i^{(3a)} = u \wedge x_j^{(3a-2)} = v \wedge x_j^{(3a-1)} = w} \, \mathrm{d}u\,\mathrm{d}v\,\mathrm{d}w \\
	\end{eqnarray*}
	%\begin{eqnarray*}
		%& \geq & 1 - (1-\rho) \iiint \pr{x_i^{(3a)} = u \wedge x_j^{(3a-2)} = v \wedge x_j^{(3a-1)} = w} \, \mathrm{d}u\,\mathrm{d}v\,\mathrm{d}w \\
		%& = & \rho.
	%\end{eqnarray*}
	If $c_1 = c_2$, then \eqref{eq:0} becomes $(x_i^{(3a)}-x_i^{(3a-1)})(x_j^{(3a-1)} - x_j^{(3a-2)})= 0$.  Thus,
	\begin{eqnarray*}
		& & \pr{E_{ij}^{(3a)} | x_i^{(3a-2)} = c_1 \,\wedge\, x_i^{(3a-1)} = c_2 \,\wedge\, c_1 = c_2} \\
		& = & \pr{x_j^{(3a-2)} \neq x_j^{(3a-1)}} \cdot \pr{x_i^{(3a)}\neq c_1} \\
		& \geq & \rho^2.
	\end{eqnarray*}
	\cancel{
	\begin{align*}
		&~~\pr{E_{ij}^{(3a)} | x_i^{(3a-2)} = c_1 \,\wedge\, x_i^{(3a-1)} = c_2 \,\wedge\, c_1 = c_2} \\
		%= &~~\pr{x_j^{(3a-2)} \not= x_j^{(3a-1)} \,\wedge\, x_i^{(3a)} \neq c_1} \\
		= &~~\pr{x_j^{(3a-2)} \neq x_j^{(3a-1)}} \cdot \pr{x_i^{(3a)}\neq c_1} \\
		\geq &~~\left(1-\pr{x_j^{(3a-2)} = x_j^{(3a-1)}}\right) \cdot \rho \\
		= &~~\rho \cdot \left(1 - \int \pr{x_j^{(3a-2)}=c} \cdot \pr{x_j^{(3a-1)} = c} \mathrm{d}c \right) \\
		\geq &~~\rho \cdot \left(1 - (1-\rho)\int \pr{x_j^{(3a-2)}=c} \mathrm{d}c \right) \\
		= &~~\rho^2.
	\end{align*}
}
	The above shows that the probability of $E_{ij}^{(3a)}$ conditioned on some fixed values of $x_i^{(3a-2)}$ and $x_i^{(3a-1)}$ is at least $\rho^2$.  Hence, $\displaystyle \pr{E_{ij}^{(3a)}} \geq \rho^2 \cdot \iint \pr{x_i^{(3a-2)} = c_1 \,\wedge\, x_i^{(3a-1)} = c_2} \mathrm{d}c_1\mathrm{d}c_2 = \rho^2$.
	
	The events in $\bigcup_{a=1}^{\ln^2 n} E_{ij}^{(3a)}$ are independent of each other.  Therefore,
	\begin{align*}
		\pr{\bigcup_{a=3}^{3\ln^2 n} E_{ij}^{(a)}} & \geq \pr{\bigcup_{a=1}^{\ln^2 n} E_{ij}^{(3a)}} = 1 - \prod_{a=1}^{\ln^2 n} \pr{\overline{E}_{ij}^{(3a)}} \geq 1 - (1-\rho^2)^{\ln^2 n}.
	\end{align*}
	Since $n \geq e^{3/\rho^2}$, we get $(1-\rho^2)^{\ln^2 n} \leq e^{-\rho^2\ln^2 n} \leq e^{-3\ln n} = n^{-3}$, establishing the lemma.
\end{proof}

%We keep a dictionary that stores $(i,j,b_{ij})$ for all $i \neq j$ and $i,j \in [1,n]$ such that the distributions of $x_i$ and $x_j$ are non-degenerate.  Initially, $b_{ij} = 1$ for all $(i,j)$.  For each $I_a$ where $a \in [3,3\ln^2 n]$, we perform the following.  For every $(i,j)$, we check the event $E_{ij}^{(a)}$ in $O(1)$ time, set a bit variable $\beta = 0$ if $E_{ij}^{(a)}$ occurs and $\beta = 1$ otherwise, and then update $b_{ij} := b_{ij} \wedge \beta$.  After going through all $3\ln^2 n$ input instances, we put $x_i$ and $x_j$ in the same class if and only if $b_{ij} = 1$.  
By Lemmas~\ref{lem:degenerate} and~\ref{lem:collinear} and the union bound, we can generate the classes based on collinearity in $O(n^2\log^3 n)$ time.  The classification is correct with probability at least $1-1/n$.  We label the classes as $S_1$, $S_2$ and so on.  We use $g$ to denote the number of classes identified.

\begin{lemma}
	\label{lem:train-1}
	Assume that $n \geq e^{3/\rho^2}$.  Using $3\ln^2 n$ input instances, we can correctly identify all linear classes in $O(n^2\log^3 n)$ time and $O(n\log^2 n)$ space with probability at least $1 - 1/n$.
\end{lemma}

\subsubsection{Structures for the operation phase}

In addition to learning the linear classes, we need to construct a data structure in the training phase that will allow the operation phase to run efficiently.  We first give an overview of what this data structure will do.

The construction and operation of this data structure require the determination of a \emph{$V$-list} of real numbers $v_0 < v_1 < v_2 < \ldots < v_n < v_{n+1}$, where $v_0$ and $v_{n+1}$ denote $-\infty$ and $\infty$, respectively.  They divide the real line into $n+1$ intervals: 
\[
[v_0,v_1), [v_1, v_2), \ldots, [v_{n-1},v_n), [v_n,v_{n+1}),
\]
where we use $[v_0,v_1)$ to denote $(-\infty,v_1)$.   For every input instance $I = (x_1,x_2,\ldots, x_n)$ in the operation phase, the data structure supports the following three operations.   
\begin{itemize}
	\item[F1:] For every class $S_k$, retrieve the sorted order of the numbers in $I$ with indices in $S_k$.  Denote this sorted order as $\sigma_k$.
	
	\item[F2:] For every class $S_k$, every $i \in S_k$, and every number $x_i \in I$, determine the largest $v_r$ in the $V$-list that is less than or equal to $x_i$.
	
	\item[F3:] For every interval $[v_r,v_{r+1})$, compute a list of sorted lists $Z_r = \{ \sigma_k \cap [v_r,v_{r+1}) : k \in [1,g] \wedge \sigma_k \cap [v_r,v_{r+1}) \not= \emptyset \}$.
\end{itemize}	
We describe how to compute the $V$-list and the data structure in the following.

\paragraph{$\pmb{V}$-list.} The determination of the $V$-list requires taking another $\ln n$ input instances.  Sort all numbers in these instances into one sorted list $L$.  Then, for $i \in [1,n]$, $v_i$ in the $V$-list is the number of rank $i\ln n$ in $L$.  Note that if the distribution of $x_i$ is degenerate, the same $x_i$ appears $\ln n$ times in the sorted list $L$, which implies that $x_i$ must be selected to be an element of the $V$-list.

\paragraph{Data structure.}  The $V$-list induces $n$ horizontal lines at $y$-coordinates $v_1,v_2,\cdots,v_n$.  The data structure is based on the following arrangements of lines and their refinement into vertical slabs.
\begin{itemize}  
	
\item For each class $S_k$, fix an arbitrary index $s_k \in S_k$.  For each $i \in S_k$, we associate with $i$ the equation of the line $\ell_i$ that expresses $x_i$ as a linear function in $x_{s_k}$.  This can be done by computing the equation of the support line through $(x_{s_k}^{(a)},x_{i}^{(a)})$ and $(x_{s_k}^{(b)},x_{i}^{(b)})$ for two arbitrary, distinct input instances $I_a$ and $I_b$ in $O(1)$ time.  The total processing time over all classes is $O(n)$.

\item For every class $S_k$, let $A_k$ be the arrangement formed by the $n$ horizontal lines induced by $v_1,v_2,\ldots,v_n$ and the lines $\ell_i$'s for all $i \in S_k$.  The size of $A_k$ is $O(n|S_k|)$.

\item Draw vertical lines through the vertices of $A_k$.  Two adjacent vertical lines bound a vertical slab.  Denote by $W_k$ the set of slabs obtained.  The size of $W_k$ is $O(n|S_k|)$.  Within each slab in $W_k$, each line $\ell_i$ in $A_k$ lies between two consecutive values $v_r$ and $v_{r+1}$, i.e., $v_r$ is the predecessor of $\ell_i$ in the $V$-list.  Moreover, the bottom-to-top order of the lines for $S_k$ is fixed within a slab.

\end{itemize} 
We compute $A_k$ and store $W_k$ as a collection of ordered lists of lines as follows.
\begin{enumerate}
	\item Compute $A_k$ by a plane sweep in $O(n|S_k|\log n)$ time.
	
	\item Each slab in $W_k$ is represented as a list of lines for $S_k$ ordered from bottom to top.  Each line $\ell_i$ is associated with its predecessor $v_r$ in the $V$-list within the slab.  These ordered lists of lines for $W_k$ are stored in a persistent search tree~\cite{driscoll89} in order to save storage and processing time.  A persistent search tree is a collection of balanced search trees of different versions.  Given a tree of a specific version, it can be searched in logarithmic time.  When the first version is constructed, it is just an ordinary balanced search tree.  When an update (including insertion, deletion and changing the content of a node) on the current version is specified, instead of modifying the current version, a new version is generated that incorporates the update.  Each update uses $O(1)$ extra amortized space and takes logarithmic time.  The construction of the persistent search tree for $W_k$ is done as follows.
	
	\item Initialize the first version of the search tree to store the lines for $S_k$ in the leftmost slab of $W_k$ in decreasing order of their slopes (which is the same as the bottom-to-top order).  Lines with positive slopes are labelled with $v_0$ as their predecessors in this slab.  Similarly, lines with negative slopes are labelled with $v_n$.  The construction of this version takes $O(|S_k|\log |S_k|)$ time and $O(|S_k|)$ space.  Run a plane sweep over $A_k$ from left to right.  We exit the current slab and enter a new slab when crossing a vertex of $A_k$.  If we cross an intersection between two lines $\ell_i$ and $\ell_j$, then we swap $\ell_i$ and $\ell_j$ in the persistent search tree (by swapping node contents).  Suppose that we cross an intersection between a horizontal line $y = v_r$ and a line $\ell_i$.  If $\ell_i$ is above $y=v_r$ to the right of this intersection, then we update the predecessor of $\ell_i$ to $v_r$; otherwise, we update the predecessor of $\ell_i$ to $v_{r-1}$.  As a result, we obtain a new version of the persistent search tree in $O(\log |S_k|)$ time and $O(1)$ extra amortized space.  Constructing all versions thus take $O(n|S_k|\log |S_k|)$ time and $O(n|S_k|)$ space.  Notice that there is one version for each slab in $W_k$.
	
	\item Given an input instance $I$ in the operation phase, we need to provide fast access to different versions of the persistent search tree for all classes.  This is done as follows.
	
	\begin{enumerate}
		
		\item Take another $n^{\eps}$ input instances for any choice of $\eps \in (0,1)$.  For every class $S_k$, record the frequencies of $x_{s_k}$ falling into the slabs in $W_k$ among these $n^\eps$ instances (via binary search among the slabs). This step takes $O(\sum_{k=1}^g n|S_k| + n^{1+\eps}\log n) = O(n^2)$ total time over all classes.  Then, for every class $S_k$, we build a binary search tree $T_k$ on these slabs whose expected search time is asymptotically optimal with respect to the recorded frequencies.  Each $T_k$ has $O(n|S_k|)$ nodes and can be constructed in $O(n|S_k|)$ time~\cite{fredman75,mehlhorn75}.
		
		\item Each node in $T_k$ corresponds to a slab in $W_k$.  We associate with this node a pointer to the version of the persistent search tree for the corresponding slab.  A very low frequency cannot give a good estimate of the probability distribution of $x_{s_k}$, so navigating down $T_k$ to a node of very low frequency may be too time-consuming.  Thus, if a search of $T_k$ reaches a node at depth below $\frac{\eps}{3}\log_2 n$, we answer the query by performing a binary search among the slabs in $W_k$, which takes $O(\log n)$ time.  Note that the slab also stores a pointer to the corresponding version of the persistent search tree.
		
	\end{enumerate}
\end{enumerate}
We explain how to use the data structure to support the operations F1, F2 and F3 described earlier.  

Let $I = (x_1,x_2,\ldots,x_n)$ be an input instance in the operation phase.  For every class $S_k$, we query $T_k$ with $x_{s_k}$ to find the slab in $W_k$ whose span of $x$-coordinates contains $x_{s_k}$.  This provides access to the version of the persistent search tree for that slab.  Denote this version by~$T$.  An inorder traversal of $T$ gives the sorted order of the lines $\ell_i$'s for all $i \in S_k$ in $O(|S_k|)$ time.  Each line $\ell_i$ stores its predecessor $v_r$ in the $V$-list.  The above handles F1 and F2.  Consider F3.  For $k = 1, 2, \cdots, g$, we walk through the sorted list of lines $\ell_i$'s in $S_k$ produced by the inorder traversal of $T$, and for each $\ell_i$ encountered in the traversal, let $v_r$ be the predecessor of $\ell_i$, and we append $x_i$ to the list in $Z_r$ under construction, i.e., the list that represents $\sigma_k \cap [v_r,v_{r+1})$.  Afterwards, we scan all intervals and output $\sigma_k \cap [v_r,v_{r+1})$ for all $k$ and $r$.

We summarize the above processing in the following result.
	
\begin{lemma}
	\label{lem:train-2}
	Assume that the hidden classes $S_1, S_2, \ldots, S_g$ have been determined.
	\begin{emromani}
		
		\item Using $\ln n$ input instances, we can set the $V$-list $(v_0,v_1,\ldots,v_n,v_{n+1})$ in $O(n\log^2 n)$ time using $O(n\log n)$ space, where $v_0 = -\infty$, $v_{n+1} = \infty$, and for $i \in [1,n]$, $v_i$ is the number of rank $i\ln n$ in the sorted list of all numbers in the $\ln n$ input instances.
		
		\item Given the $V$-list, there is a data structure that performs functions F1, F2, and F3 in $O(E + n)$ expected time for every input instance in the operation phase, where $E$ is the total expected time to query the $T_k$'s.  The data structure uses $O(n^2)$ space and can be constructed in $O(n^2\log n)$ time using $n^\eps$ input instances.
	\end{emromani}
\end{lemma}

\subsection{Operation phase}

Given an input instance $I = (x_1,\cdots,x_n)$, the operation phase proceeds as follows.
\begin{enumerate}
	
	\item During the construction of the $V$-list in the training phase, for each $x_i$ that is degenerately distributed, $x_i$ must appear $\ln n$ times when we sort the concatenation of $\ln n$ input instances. Therefore, for each degenerately distributed $x_i$, there is a unique $v_r$ in the $V$-list that is equal to $x_i$, and we mark $v_r$.
	
	\item Use Lemma~\ref{lem:train-2}(ii) to determine for every class $S_k$, the sorted sequence $\sigma_k$ of numbers belonging to $S_k$ and for every interval $[v_r,v_{r+1})$, the list of sorted lists $Z_r = \{ \sigma_k \cap [v_r,v_{r+1}) : k \in [1,g] \wedge \sigma_k \cap [v_r,v_{r+1}) \not= \emptyset \}$.  Note that $|Z_r| \leq g$.

	%\item For every interval $[v_r,v_{r+1})$, initialize an empty set $Z_r$ of lists.   For each $k \in [1,g]$, if $\sigma_k \cap [v_r,v_{r+1})$ is non-empty, add $\sigma_k \cap [v_r,v_{r+1})$ to $Z_r$.  Using the predecessors of the $x_i$'s in the class $S_k$, we can distribute $\sigma_k$ among the $Z_r$'s in $O(|\sigma_k|) = O(|S_k|)$ time.  Note that $|Z_r| \leq g$.
	
	\item For every interval $[v_r,v_{r+1})$, merge all lists in $Z_r$ into one sorted list.  The merging is facilitated by a min-heap that stores the next element from each list in $Z_r$.  Thus, each step of the merging takes $O(\log |Z_r|)$ time.  
	
	\item Finally, we concatenate in $O(n)$ time the marked $v_r$'s and the merged lists for all $Z_r$'s to form the output sorted list.
	
\end{enumerate}

Correctness is obvious.  The limiting complexity has two main components.   First, the sum of expected query times of all $T_k$'s in Lemma~\ref{lem:train-2}(ii).  Second, the total time spent on merging the lists in $Z_r$ for $r \in [0,n]$.  The remaining processing time is $O(n + \sum_{k=1}^g |S_k|) = O(n)$.  We give the analysis in the next section to show that the first two components sum to $O(n/\eps + H_\pi/\eps)$.  Recall that $\pi(I)$ is the sequence of the ranks of numbers in $I$, which is a permutation of $[n]$, and $H_\pi$ is the entropy of the distribution of $\pi(I)$.

\subsection{Analysis}
\label{sec:analysis-1}

Assign labels 0 to $n+1$ to $v_0,v_1, \cdots, v_n,v_{n+1}$ in this order.  Similarly, assign labels $n+2$ to $2n+1$ to the input numbers $x_1,\cdots,x_n$ in this order.

Define the random variable $B^V$ to be the permutation of the labels that appear from left to right after sorting $\{v_0,\cdots,v_{n+1}\} \cup \{x_1,\cdots,x_n\}$ in increasing order.

For each $k \in [1,g]$, define a random variable $B_k^V$ to be the permutation of the labels that appear from left to right after performing the following operations: (1)~sort $\{v_0,\cdots,v_{n+1}\} \cup \{x_i : i \in S_k\}$ in increasing order, and (2)~remove all $v_r$'s that do not immediately precede some $x_i$'s in the sorted list.  Let $H_k^V$ denote the entropy of the distribution of $B_k^V$.  Determining $B^V_k$ takes at least $H_k^V$ expected time by Shannon's theory~\cite{cover06}.

%For each $k \in [1,g]$, define a random variable $\hat{B}_k^V$ to be the label of the predecessor of $x_{s_k}$ in the $V$-list.  Let $\hat{H}^V_k$ denote the entropy of the distribution of $\hat{B}^V_k$.  Determining the predecessor of $x_{s_k}$ in the $V$-list takes at least $\hat{H}^V_k$ expected time.

Our algorithm uses Lemma~\ref{lem:train-2}(ii) to  
%queries $T_k$ 
%and $\hat{T}_k$ 
%for $k \in [1,g]$, 
construct $\sigma_k \cap [v_r,v_{r+1})$ for all $k$ and $r$ in $O(E + n)$ expected time, where $E$ is the total expected time to query the $T_k$'s.
Then, it performs mergings in $O(\sum_{r=0}^n \sum_{k=1}^g |\sigma_k \cap [v_r,v_{r+1})| \log |Z_r|)$ time. Recall that $|Z_r|$ is the number of classes that have numbers falling into $[v_r,v_{r+1})$.  As shown in Lemma~3.4 in~\cite{ailon11} and the discussion that immediately follows its proof, the expected query complexity of $T_k$ is $O(H^V_k/\eps)$.  The limiting complexity is thus equal to
\begin{equation}
O\left(n + \frac{1}{\eps}\sum_{k=1}^g H^V_k\right) + 
O\left(\mathrm{E}\left[\sum_{r=0}^n \sum_{k=1}^g |\sigma_k \cap [v_r,v_{r+1})| \log |Z_r|\right]\right).
\label{eq:-1}
\end{equation}

We bound $\sum_{k=1}^g H^V_k$ and
$\mathrm{E}\left[\sum_{r=0}^n \sum_{k=1}^g |\sigma_k \cap [v_r,v_{r+1})| \log |Z_r|\right]$ in the rest of this section.  We need two technical results. 

\begin{lemma}{{\em \cite[Theorem~2.39]{ray}}}
	\label{lem:joint}
	Let $H(X_1,\cdots,X_n)$ be the joint entropy of independent random variables $X_1,\cdots,X_n$.  Then $H(X_1,\cdots,X_n) = \sum_{i=1}^n H(X_i)$.
\end{lemma}

\begin{lemma}{{\em \cite[Lemma~2.3]{ailon11}}}
	\label{lem:relate}
	Let $X : {\cal U} \rightarrow {\cal X}$ and $Y : {\cal U} \rightarrow {\cal Y}$ be two random variables obtained with respect to the same arbitrary distribution over the universe $\cal U$.  Suppose that the function $f: (I,X(I)) \mapsto Y(I)$, $I \in {\cal U}$, can be computed by a comparison-based algorithm with $C$ expected comparisons, where the expectation is over the distribution on $\cal U$.  Then, $H(Y) \leq C + O(H(X))$.
\end{lemma}

We show that $\sum_{k=1}^g H^V_k = O(n + H_\pi)$.

\begin{lemma}
	\label{lem:entropies}
	$\sum_{k=1}^g H^V_k = O\left(n + H(B^V)\right) = O\left(n + H_\pi \right)$.
	\cancel{
	    \hspace{.2in}
		\begin{emromani}
			\item $\sum_{k=1}^g H^V_k = O\left(n + H(B^V)\right) = O\left(n + H_\pi \right)$,
			\item $\sum_{k=1}^g \hat{H}^V_k = O(n + H_\pi)$.
		\end{emromani}
	}
\end{lemma}
\begin{proof}
	%Consider (i).  
	Suppose that we are given a setting of $B^V$, i.e., the permutation of labels from left to right in the sorted order of $\{v_0,\cdots,v_{n+1}\} \cup \{x_1,\cdots,x_n\}$.  We scan the sorted list from left to right.  We maintain the most recently scanned $v_r$.  Suppose that we see a number $x_i$.  Let $S_k$ be the class to which $x_i$ belongs.  If this is the first time that we encounter an index in $S_k$ after seeing $v_r$, we initialize an output list for $B_k^V$ that contains the label of $v_r$ followed by the label of $x_i$.  If this is not the first time that we encounter an index in $S_k$ after seeing $v_r$, we append the label of $x_i$ to the output list for $B_k^V$.  Clearly, we obtain the settings of all $B_k^V$'s correctly from $B^V$.  The number of comparisons needed is $O(n)$.  Therefore, Lemmas~\ref{lem:joint} and~\ref{lem:relate} imply that $\sum_{k=1}^g H^V_k = H(B_1^V,\cdots,B_g^V) = O(n + H(B^V))$.
	
	Given $(I,\pi(I))$, we use $\pi(I)$ to sort $I$ and then merge the sorted order with $(v_0,\cdots,v_{n+1})$.  Afterwards, we scan the sorted list to output the labels of the numbers.  This gives the setting of $B^V$.  Clearly, $O(n)$ comparisons suffice, and so Lemma~\ref{lem:relate} implies that $H(B^V) = O(n + H_\pi)$.  
	%This completes the proof of (a).
	%
	%The settings of $\hat{B}_1^V, \cdots, \hat{B}_g^V$ can be derived similarly by using $\pi(I)$ to sort $I$, merging the sorted sequence with $(v_0,\cdots,v_{n+1})$, and then scanning the merged sequence.  Then, Lemmas~\ref{lem:joint} and~\ref{lem:relate} imply that $\sum_{k=1}^g \hat{H}^V_k = H(\hat{B}_1^V,\cdots,\hat{B}_g^V) = O(n + H_\pi)$, establishing (b).
\end{proof}

\cancel{
We will show that it holds with high probability that $\mathrm{E}[|Z_r|] = O(1)$ for all $r \in [0,n]$ simultaneously.  
%By this result and Jensen's inequality~\cite{inequality}, it holds with high probability that $\mathrm{E}\left[\log |Z_r|\right] \leq \log\left(\mathrm{E}\left[|Z_r|\right]\right) = O(1)$ for all $r \in [0,n]$ simultaneously.  Therefore, 
It implies that $\mathrm{E}\left[\max_{r \in [0,n]} |Z_r|\right] = O(1)$ with high probability.  Then,
\begin{align*}
\mathrm{E}\left[\sum_{r=0}^n |\sigma_k \cap [v_r,v_{r+1})| \log |Z_r|\right] 
& \leq \mathrm{E}\left[\max_{r \in [0,n]} |Z_r|  \cdot \sum_{r=0}^n |\sigma_k \cap [v_r,v_{r+1})| \right]  \\
& = |\sigma_k| \cdot \mathrm{E}\left[\max_{r \in [0,n]} |Z_r|  \right] \\
& = O(|\sigma_k|).
\end{align*}
Hence,
\begin{align*}
	\mathrm{E}\left[\sum_{r=0}^n \sum_{k=1}^g |\sigma_k \cap [v_r,v_{r+1})| \log |Z_r|\right] 
	& = \sum_{k=1}^g \mathrm{E}\left[ \sum_{r=0}^n |\sigma_k \cap [v_r,v_{r+1})| \log |Z_r|\right] \\
	& \leq O\left(\sum_{k=1}^g |\sigma_k|\right) \\
	& = O(n).
\end{align*}
The second term in \eqref{eq:-1} can thus be replaced by $O(n)$.
}

Lemma~\ref{lem:entropies} takes care of the first term in \eqref{eq:-1}.  We will show that the second term in \eqref{eq:-1} is $O(n)$ with high probability.  We first prove that $\mathrm{E}[|Z_r|] = O(1)$ for all $r \in [0,n]$ with high probability.  Our proof is modeled after the proof of a similar result in~\cite{ailon11}.  There is a small twist due to the handling of the classification.

\begin{lemma}
	\label{lem:Z}
	It holds with probability at least $1-1/n$ that for all $r \in [0,n]$, $\mathrm{E}[|Z_r|] = O(1)$.
\end{lemma}
\begin{proof}
	Let $I_1, \cdots, I_{\ln n}$ denote the input instances used in the training phase for building the $V$-list.  Let $y_1,y_2,\cdots,y_{n\ln n}$ denote the sequence formed by concatenating $I_1, \cdots, I_{\ln n}$ in this order. We adopt the notation that for each $\alpha \in [1,n\ln n]$,  $y_\alpha$ belongs to the class $S_{k_\alpha}$ and the input instance $I_{a_\alpha}$. 
	
	Fix a pair of distinct indices $\alpha,\beta \in [1,n\ln n]$ such that $y_\alpha \leq y_\beta$.  Let ${\cal J}_{\alpha}^\beta$ be the set of index pairs $\{(a,k) : a \in [1,\ln n], k \in [1,g]\} \setminus \{(a_\alpha,k_\alpha),(a_\beta,k_\beta)\}$.   For any $(a,k) \in {\cal J}_\alpha^{\beta}$, let $Y_{\alpha}^{\beta}(a,k)$ be an indicator random variable such that if some element of the input instance $I_a$ that belongs to $S_k$ falls into $[y_\alpha,y_\beta)$, then $Y_{\alpha}^{\beta}(a,k) = 1$; otherwise, $Y_{\alpha}^{\beta}(a,k) = 0$.  Define $Y_{\alpha}^{\beta} = \sum_{(a,k) \in {\cal J}_{\alpha}^{\beta}} Y_{\alpha}^{\beta}(a,k)$.
	
	Among the $(a,k)$'s in ${\cal J}_{\alpha}^{\beta}$, the random variables $Y_{\alpha}^{\beta}(a,k)$ are independent from each other.  By Chernoff's bound, for any $\mu \in [0,1]$, 
	\[
	\pr{Y_{\alpha}^{\beta} > (1-\mu) \mathrm{E}[Y_{\alpha}^{\beta}]} > 1 - e^{-\mu^2\mathrm{E}[Y_{\alpha}^{\beta}]/2}.
	\]
	Since we take every $\ln n$ numbers in forming the $V$-list, we want to discuss the probability of $Y_{\alpha}^{\beta} > \ln n$.  This motivates us to consider $\mathrm{E}[Y_{\alpha}^{\beta}] > \ln n/(1-\mu)$.  We also want the probability bound $1 - e^{-\mu^2\mathrm{E}[Y_{\alpha}^{\beta}]/2}$ of $Y_{\alpha}^{\beta} > \ln n$ to be at least $1 - n^{-5}$.  This allows us to apply the union bound over at most $(n\ln n)(n\ln n - 1)$ choices of $\alpha$ and $\beta$ to obtain a probability bound of at least $1 - \ln^2 n/n^3$.  Therefore, as we consider $\mathrm{E}[Y^{\beta}_{\alpha}] > \ln n/(1-\mu)$, we want $1 - e^{-\mu^2\ln n/(2(1-\mu))} = 1 - n^{-\mu^2/(2(1-\mu))} = 1 - n^{-5}$.  Equivalently, we require $\mu^2/(2(1-\mu)) = 5$ which is satisfied by setting $\mu = \sqrt{35}-5 \approx 0.9161$.  We conclude that:
	\begin{quote}
	It holds with probability at least $1-\ln^2 n/n^3$ that for any pair of distinct indices $\alpha,\beta \in [1,n\ln n]$ such that $y_\alpha \leq y_\beta$, if $\mathrm{E}[Y_{\alpha}^{\beta}] > \frac{1}{6-\sqrt{35}}\ln n$, then $Y_{\alpha}^{\beta} > \ln n$.
	\end{quote}
	
	For every $r \in [0,n+1]$, let $y_{\alpha_r}$ denote $v_r$, where $y_{\alpha_0} = -\infty$ and $y_{\alpha_{n+1}} = \infty$.  Fix a particular $r \in [0,n+1]$.  By construction, there are at most $\ln n$ numbers among $I_1, \cdots, I_{\ln n}$ that fall in $[v_r,v_{r+1})$, which guarantees the event of $Y_{\alpha_r}^{\alpha_{r+1}} \leq \ln n$.  Our previous conclusion implies that $\mathrm{E}[Y_{\alpha_r}^{\alpha_{r+1}}] \leq \frac{1}{6-\sqrt{35}}\ln n$ with probability at least $1-\ln^2 n /n^3$.  
	
	We relate $\mathrm{E}[Y_{\alpha_r}^{\alpha_{r+1}}]$ to $\mathrm{E}[|Z_r|]$ as follows.  Let $X_{kr}$ be an indicator random variable such that if some element of the input instance that belongs to $S_k$ falls into $[v_r,v_{r+1})$, then $X_{kr} = 1$; otherwise, $X_{kr} = 0$.  Then $\sum_{k=1}^g X_{kr} = |Z_r|$, implying that $\sum_{k=1}^g \mathrm{E}[X_{kr}] = \mathrm{E}[|Z_r|]$.  The random process that generates the input instances is independent of the training phase.  It follows that %$\mathrm{E}[Y_{\alpha_r}^{\alpha_{r+1}}]$ is almost the same as $\sum_{a=1}^{\ln n} \sum_{k=1}^g \mathrm{E}[X_{kr}]$, except that the index pairs $(a_{\alpha_r},k_{\alpha_r})$ and $(a_{\alpha_{r+1}},k_{\alpha_{r+1}})$ are excluded from ${\cal J}_{\alpha_r}^{\alpha_{r+1}}$ but these two cases are considered in $\sum_{a=1}^{\ln n} \sum_{k=1}^g \mathrm{E}[X_{kr}]$.  Therefore,
	\begin{equation}
	\mathrm{E}[Y_{\alpha_r}^{\alpha_{r+1}}] \geq \left(\sum_{a=1}^{\ln n}\sum_{k=1}^g \mathrm{E}[X_{kr}] \right) - 2 = \ln n \cdot \mathrm{E}[|Z_r|] - 2
	\label{eq:1}
	\end{equation}
	because the index pairs $(a_{\alpha_r},k_{\alpha_r})$ and $(a_{\alpha_{r+1}},k_{\alpha_{r+1}})$ are excluded from ${\cal J}_{\alpha_r}^{\alpha_{r+1}}$ but they are considered in $\sum_{a=1}^{\ln n} \sum_{k=1}^g \mathrm{E}[X_{kr}]$.  
	
	We have shown previously that $\mathrm{E}[Y_{\alpha_r}^{\alpha_{r+1}}] \leq \frac{1}{6-\sqrt{35}}\ln n$ with probability at least $1-\ln^2 n/n^3$.  It follows that $\mathrm{E}[|Z_r|] = O(1)$ with probability at least $1-\ln^2 n/n^3$.  Since the above statement holds for every fixed $r \in [0,n]$, by the union bound, it holds with probability at least $1-1/n$ that $\mathrm{E}[|Z_r|] = O(1)$ for all $r \in [0,n]$.
\end{proof}

We are ready to bound the second term in \eqref{eq:-1}.  

\begin{lemma}
	\label{lem:2}
	 It holds with probability at least $1 - 1/n$ that \[\mathrm{E}\left[\sum_{k=1}^g \sum_{r=0}^n |\sigma_k \cap [v_r,v_{r+1})| \log |Z_r| \right] = O(n).\]
\end{lemma}
\begin{proof}
Let $n_{kr}$ denote $|\sigma_k \cap [v_r,r_{r+1})|$.  Let $z_r$ denote $|Z_r|$.  The largest possible values of $n_{kr}$ and $z_r$ are $n$ and $g$, respectively.
\[
	\mathrm{E}\left[\sum_{k=1}^g \sum_{r=0}^n n_{kr}\log z_r\right] 
	\leq \sum_{k=1}^g \sum_{r=0}^n \mathrm{E}\left[n_{kr}z_r\right] \\
	= \sum_{k=1}^g \sum_{r=0}^n \sum_{i=0}^{gn} i \cdot \pr{n_{kr}z_r = i}. 	
\]
The range of $i$ can be reduced to $[1,gn]$ without changing the sum:
\[
\sum_{i=0}^{gn} i \cdot \pr{n_{kr}z_r = i} =
\sum_{i=1}^{gn} i \cdot \pr{n_{kr}z_r = i} = 
\sum_{j=1}^{g} \sum_{l=1}^{n} jl \cdot \pr{z_r = j \wedge n_{kr} = l}.
\]
The last equality follows from the fact that if $j \not= j'$ or $l \not= l'$, then the events $z_r = j \wedge n_{kr} = l$ and $z_r = j' \wedge n_{kr} = l'$ are disjoint.

Let $y_{kr}$ be a random variable that counts the number of classes other than $S_k$ that have numbers in $[v_r,v_{r+1})$. In the event of $n_{kr} = l$ for some $l \in [1,n]$, the class $S_k$ has number(s) in $[v_r,v_{r+1})$, implying that $z_r = y_{kr} + 1$.  Therefore,
\begin{eqnarray*}
\sum_{j=1}^g
\sum_{l=1}^n jl \cdot \pr{z_r = j \wedge n_{kr} = l} 
& = & 
\sum_{j=0}^{g-1}
\sum_{l=1}^n (j+1)l \cdot \pr{y_{kr} = j \wedge n_{kr} = l} \\
& = & 
\sum_{j=0}^{g-1}
\sum_{l=1}^n (j+1)l \cdot \pr{y_{kr} = j}\cdot \pr{n_{kr} = l}.
\end{eqnarray*}
In the last step, the equality of $\pr{y_{kr} = j \wedge n_{kr} = l}$ and $\pr{y_{kr} = j}\cdot \pr{n_{kr} = l}$ follows from the independence of the events $y_{kr} = j$ and $n_{kr} = l$.  Hence,
\begin{eqnarray*}
	\mathrm{E}\left[\sum_{k=1}^g \sum_{r=0}^n n_{kr}\log z_r\right] & \leq  &
	\sum_{k=1}^g \sum_{r=0}^n \sum_{j=0}^{g-1} \sum_{l=1}^n (j+1)l \cdot \pr{y_{kr} = j} \cdot \pr{n_{kr} = l} \\
	& = & \sum_{k=1}^g \sum_{r=0}^n \sum_{j=0}^{g-1} (j+1) \cdot \pr{y_{kr} = j} \cdot \sum_{l=1}^n l \cdot \pr{n_{kr} = l} \\
	& = & 
	\sum_{k=1}^g \sum_{r=0}^n \sum_{j=0}^{g-1} \mathrm{E}[n_{kr}] \cdot (j+1) \cdot \pr{y_{kr} = j} \\
	& = & 
	\sum_{k=1}^g \sum_{r=0}^n \mathrm{E}[n_{kr}] \cdot \left(\sum_{j=0}^{g-1} j \cdot \pr{y_{kr} = j}  + \sum_{j=0}^{g-1} \pr{y_{kr}=j} \right)\\
	&= & 
	\sum_{k=1}^g \sum_{r=0}^n \mathrm{E}[n_{kr}] \cdot \left(\mathrm{E}[y_{kr}] + 1\right).
\end{eqnarray*}
For all $k \in [1,g]$, $z_r \geq y_{kr}$ by their definitions, and so $\mathrm{E}[z_r] \geq \mathrm{E}[y_{kr}]$.  By Lemma~\ref{lem:Z}, it holds with probability at least $1 - 1/n$ that $\mathrm{E}[y_{kr}] + 1 = O(1)$ for every $k \in [1,g]$ and every $r \in [0,n]$.  Finally,
\[
	\mathrm{E}\left[\sum_{k=1}^g \sum_{r=0}^n n_{kr}\log z_r\right] 
	\leq 
	O\left(\sum_{k=1}^g \sum_{r=0}^n \mathrm{E}[n_{kr}] \right) =
    O\left(\mathrm{E}\left[\sum_{k=1}^g \sum_{r=0}^n n_{kr}\right] \right) 
	= O(n).
\]
\end{proof}

By \eqref{eq:-1} and Lemmas~\ref{lem:entropies} and~\ref{lem:2}, we conclude that the limiting complexity of the sorter is $O(n/\eps + H_\pi/\eps)$ as stated in Theorem~\ref{thm:1}.  The $O(n^2)$ space needed by the operation phase follows from Lemma~\ref{lem:train-2}(ii).  In the training phase, the space usage, the number of input instances, and the processing time required follow from Lemmas~\ref{lem:train-1} and~\ref{lem:train-2}.  The success probability of $1 - 1/n$ follows from Lemma~\ref{lem:2}.  This completes the proof of Theorem~\ref{thm:1}.

\cancel{

\begin{theorem}
	\label{thm:1}
	For any $\eps \in (0,1)$, there exists a self-improving sorter of $O(n/\eps + H_\pi/\eps)$ limiting complexity for any product distribution with hidden linear classes.  The storage needed is $O(n^2)$.  The training phase processes $O(n^{\eps})$ input instances in $O(n^2\log^3 n)$ time, and it succeeds with probability at least $1-1/n$.
\end{theorem}

}

\section{Mixture of product distributions}

Let $\kappa$ be the number of product distributions in the mixture.  Although $\kappa$ is hidden, we are given an upper bound $m$ of $\kappa$.  Let ${\cal D}_q$, $q \in [1,\kappa]$, denote the hidden product distributions in the mixture.  The input distribution is $\sum_{q=1}^\kappa \lambda_q {\cal D}_q$ for some hidden positive $\lambda_q$'s such that $\sum_{q=1}^\kappa \lambda_q = 1$.  
%In other words, over a large collection of input instances, a fraction $\lambda_q$ of them are expected to be drawn from ${\cal D}_q$ for each $q \in [1,\kappa]$.  There is now dependency among the input numbers because they are drawn from the same ${\cal D}_q$.

\subsection{Training phase}
\label{sec:train}

Take $mn\ln(mn)$ input instances.  Denote them as $I_1, I_2, \ldots, I_{mn\ln(mn)}$.  For $a \in [1,mn\ln(mn)]$, let $x_i^{(a)}$ denote $x_i$ in $I_a$.  For every $i \in [1,n]$ and every $a \in [(i-1)m\ln(mn)+1,im\ln(mn)]$, define 
\[
s_a = x_{i}^{(a)}.
\] 
That is, we take $x_1$'s in $I_1,\ldots,I_{m\ln(mn)}$ to be $s_1,\ldots,s_{m\ln(mn)}$,  $x_2$'s in $I_{m\ln(mn)+1},\ldots,I_{2m\ln(mn)}$ to be $s_{m\ln(mn)+1},\ldots,s_{2m\ln(mn)}$, and so on.  

Sort $(s_1,s_2,\ldots,s_{mn\ln(mn)})$ in increasing order.  For $i \in [1,mn]$, define $v_i$ to be the number of rank $i\ln(mn)$ in the sorted list.  Then, construct the $V$-list $(v_0,v_1,\ldots,v_{mn},v_{mn+1})$, where $v_0 = -\infty$ and $v_{mn+1} = \infty$.  This step takes $O(mn\log^2(mn))$ time.  The $V$-list induces $mn+1$ intervals: $(-\infty,v_1)$, $[v_1,v_2)$, $\cdots$, $[v_{mn},\infty)$.  We will abuse the notation slightly to take $[v_0,v_1)$ to mean $(-\infty,v_1)$.  

To facilitate the operation phase, we group the $mn+1$ intervals into $n$ buckets as follows.  We group the first $m$ intervals into the first bucket, the next $m$ intervals into the second bucket, and so on.  There are $n$ buckets.  Each bucket contains $m$ intervals except for the last one which contains $m+1$ intervals.  Each interval keeps a pointer to the bucket that contains it.  Also, each bucket is associated with an initially empty van Emde Boas tree\cite{v77} with the intervals in that bucket as the universe.  Each tree has $O(m)$ size and can be initialized in $O(m)$ time.\footnote{The space usage according to the description in~\cite{v77} is $O(m\log m)$, but it can be improved to $O(m)$ as mentioned in~\cite{italiano}.} 

%We organize a balanced binary search tree $T^V$ whose nodes correspond to intervals in $V$.

Use another $O(m^\eps n^\eps)$ input instances to record the frequency $f_{ir}$ of $x_i$ falling into $[v_r,v_{r+1})$.  The frequencies are determined by locating the numbers in these $O(m^\eps n^\eps)$ input instances among the intervals using binary search.  The total time needed is $O(m^\eps n^{1+\eps}\log (mn))$.  Then, for every $i \in [1,n]$, build an asymptotically optimal binary search tree $T_i$ with respect to the $f_{ir}$'s on the intervals with positive frequencies.  Each $T_i$ has $O(m^\eps n^\eps)$ size and can be constructed in $O(m^\eps n^\eps)$ time~\cite{fredman75,mehlhorn75}.  If a search of $T_i$ reaches a node at depth below $\frac{\eps}{3}\log_2 (mn)$ or is unsuccessful, we answer the query by performing a binary search among the $mn+1$ intervals in $O(\log (mn))$ time.

Let $P_i$ be a random variable indicating the predecessor of $x_i$ in the $V$-list.  Let $H(P_i)$ denote the entropy of the distribution of $P_i$.  As shown in~\cite[Lemma~3.4]{ailon11}, querying $T_i$ takes $O(H(P_i)/\eps)$ expected time (including the binary search among the $mn+1$ intervals, if applicable).

We summarize the processing in the training phase in the following result.

\begin{lemma}
	\label{lem:train-3}
	The training phase constructs the following structures.
	\begin{emromani}
		\item The $V$-list $(v_0,v_1,\ldots,v_{mn+1})$ is constructed in $O(mn\log^2(mn))$ time using $mn\ln(mn)$ input instances and $O(mn\log(mn))$ space, where $v_0 = -\infty$, $v_{mn+1} = \infty$, and for $i \in [1,mn]$, $v_i$ is the number of rank $i\ln(mn)$ in $\bigcup_{i=1}^n \{ x_i^{(a)}: a \in [(i-1)m\ln(mn)+1,im\ln(mn)]\}$.
		
		\item The $mn+1$ intervals induced by the $V$-list are organized as $n$ consecutive buckets of $m$ intervals each, except for the last bucket which contains $m+1$ intervals.  Each bucket keeps an initially empty van Emde Boas tree with the intervals in that bucket as the universe.  The processing time and space needed are $O(mn)$.
		
		\item  Search trees $T_i$ for $i \in [1,n]$ are built on the intervals $[v_0,v_1),\ldots,[v_{mn},v_{mn+1})$ using $O(m^\eps n^\eps)$ input instances.  The processing time is $O(m^\eps n^{1+\eps}\log(mn))$ and the search trees use $O(m^\eps n^{1+\eps})$ space.  For any input instance $(x_1,\ldots,x_n)$ in the operation phase, $T_i$ can be queried to find the interval that contains $x_i$ in  $O(H(P_i)/\eps)$ expected time.
	\end{emromani}
\end{lemma}

\subsection{Operation phase}
\label{sec:oper}

Given an input instance $I = (x_1,\cdots, x_n)$, for each $i \in [1,n]$, we search $T_i$ to place $x_i$ in the interval $[v_r,v_{r+1})$ that contains it. For each $r \in [0,mn]$, the interval $[v_r,v_{r+1})$ keeps a list $N_r$ of $x_i$'s that fall into it.  We sort each $N_r$ in $O(|N_r|\log |N_r|)$ time.  
Recall that querying $T_i$ with $x_i$ takes $O(H(P_i)/\eps)$ expected time, where $P_i$ is the random variable indicating the predecessor of $x_i$ in the $V$-list.  Therefore, the total time for processing $I$ is $O\left(\frac{1}{\eps}\sum_{i=1}^n H(P_i) + \mathrm{E}\left[\sum_{r=0}^{mn} |N_r|\log |N_r|\right] \right)$ plus the time to concatenate the sorted lists together.  One easy way to perform the concatenation is to scan all $mn+1$ intervals from left to right, but this takes $O(mn)$ time.  We describe an improvement below.

%We maintain a set $\cal U$ of disjoint ranges such that they cover $\mathbb{R}$ together, and each range is a union of one or more intervals in $V$.  Initially, $U = \{(-\infty,\infty)\}$.  When $x_1$ arrives, $T_1$ tells us the interval $[v_r,v_{r+1})$ that contains it.  So we split $(\infty,\infty)$ and update $U := \{(\infty,v_r), [v_r,v_{r+1}), [v_{r+1},\infty)\}$.  The list $N_r$ associated with $[v_r,v_{r+1})$ contains $x_1$ alone.  When $x_2$ arrives, $T_2$ tells the interval $[v_s, v_{s+1})$ that contains it.  We locate the range $R \in {\cal U}$ that contains $[v_s,v_{s+1})$.  If $R = [v_r,v_{r+1})$, then we add $x_2$ to $N_r$.  Otherwise, suppose that $[v_s,v_{s+1}) \subset [v_{r+1},\infty)$ and $v_s > v_{r+1}$, we split $[v_{r+1},\infty)$ and update ${\cal U} := \{ (\infty,v_r), [v_r,v_{r+1}), [v_{r+1}, v_s), [v_s,v_{s+1}), [v_{s+1},\infty)\}$, and the list $N_s$ associated with $[v_s,v_{s+1})$ contains $x_2$ alone.  After sorting the input instance, we need to restore ${\cal U}$ to $\{(-\infty,\infty)\}$ in preparation for sorting the next input instance.

%The above example illustrates that we need a data structure to maintain $\cal U$ that supports three operations: (1)~given a singleton interval $[v_r,v_{r+1})$, find the range $R \in {\cal U}$ that contains it; (2)~given a value $v_r$ and a range $R \in {\cal U}$ such that $v_r \in R$, split $R$ into $R_1 = \{y \in R : y < v_r\}$, $R_2 = [v_r,v_{r+1})$, and $R_3 = \{y \in R : y \geq v_{r+1}\}$; (3) merge two consecutive ranges in $\cal U$ into one range. 

\begin{enumerate}
	
	\item By Lemma~\ref{lem:train-3}(ii), the $mn+1$ intervals are grouped into $n$ buckets in the training phase.  For each bucket $B$, let $U_B$ denote the van Emde Boas tree for $B$ which is initially empty.  The universe for $U_B$ is the set of intervals in $B$.  We merge the $N_r$'s for the intervals within each bucket as follows.
	
	\item  For each input number $x_i$, we perform the following steps. 
	\begin{enumerate}
		
		\item Let $[v_r,v_{r+1})$ be the interval containing $x_i$ which has been located using $T_i$.  Let $B$ be the bucket pointed to by $[v_r,v_{r+1})$.
		
		\item We search for $[v_r,v_{r+1})$ in $U_B$.  If the search fails, insert $[v_r,v_{r+1})$ into $U_B$; otherwise, do nothing.  
		
		\end{enumerate}
	
	\item By now, for each bucket $B$, $U_B$ stores all non-empty intervals in $B$.  We have already discussed the sorting of each $N_r$.  We scan the $n$ buckets in left-to-right order.  For each bucket $B$ encountered, we find the minimum element in $U_B$ and then find successors in $U_B$ iteratively.  This allows us to visit the non-empty $N_r$'s in $B$ in increasing order of $r$, so we can output the sorted $N_r$'s in increasing order.  At the end, we delete all elements from $U_B$ for each bucket $B$ in preparation for sorting the next input instance.
	
	\item The total time needed is $O(n)$ plus the time for manipulating the $n$ van Emde Boas trees.  The van Emde Boas tree~\cite{v77} supports ordered dictionary operations in $O(\log\log N)$ worst-case time each, where $N$ is the size of the universe.  This is $O(\log\log m)$ time in our case.  
	
\end{enumerate}

\begin{lemma}
	\label{lem:oper}
	In the operation phase, the search trees $T_i$'s, the $V$-list, and the van~Emde~Boas trees require $O(m^\eps n^{1+\eps})$, $O(mn)$, and $O(mn)$ space, respectively.  Sorting an input instance takes $O\left(n\log\log m + \frac{1}{\eps}\sum_{i=1}^n H(P_i) + \mathrm{E}\left[\sum_{r=0}^{mn} |N_r|\log |N_r|\right] \right)$ expected time.
\end{lemma}

\subsection{Analysis}

Let $I$ be an input instance.  Let $X_{ir}$ be a random variable such that if $x_i$ falls into $[v_r,v_{r+1})$, then $X_{ir} = 1$; otherwise, $X_{ir} = 0$.  We first bound $\sum_{q=1}^\kappa \sum_{i=1}^n \pr{X_{ir} = 1 \wedge I \sim {\cal D}_q }$.

\begin{lemma}
	\label{lem:N}
	Let $I$ be an input instance.  Let $X_{ir}$ be a random variable that is $1$ if $x_i \in [v_r,v_{r+1})$ and $0$ otherwise. It holds with probability at least $1 - 1/(mn)$ that for every $r \in [0,mn]$, $\sum_{q=1}^\kappa \sum_{i=1}^n \pr{X_{ir} = 1 \wedge I \sim {\cal D}_q} = O(1/m)$.
\end{lemma}
\begin{proof}
	In building the $V$-list in the training phase, we constructed the list $(s_1,s_2,\ldots,s_{mn\ln(mn)})$ using $mn\ln(mn)$ input instances $I_1, \cdots, I_{mn\ln(mn)}$, where $s_a$ is equal to $x_i$ in $I_a$ for every $i \in [1,n]$ and every $a \in [(i-1)m\ln(mn)+1,im\ln(mn)]$.
	
	For any $\alpha,\beta \in [1,mn\ln(mn)]$ such that $s_\alpha < s_\beta$,
	let ${\cal J}_\alpha^\beta = [1,mn\ln(mn)] \setminus \{\alpha,\beta\}$.  For every $i \in {\cal J}_\alpha^\beta$, define $Y_\alpha^\beta(i) = 1$ if $s_i \in [s_\alpha,s_\beta)$ and $Y_\alpha^\beta(i) = 0$ otherwise.  Then, define $Y_\alpha^\beta = \sum_{i \in {\cal J}_\alpha^\beta} Y_\alpha^\beta(i)$.  
	
	Among all $i \in {\cal J}_\alpha^\beta$, the variables $Y_\alpha^\beta(i)$'s are independent from each other because the $s_i$'s are taken from independent input instances.  By Chernoff's bound, for any $\mu \in [0,1]$, 
	\[
	\pr{Y_{\alpha}^{\beta} > (1-\mu) \mathrm{E}[Y_{\alpha}^{\beta}]} > 1-e^{-\mu^2\mathrm{E}[Y_{\alpha}^{\beta}]/2}.
	\]
	Since we take every $\ln(mn)$ numbers in forming the $V$-list, we want to discuss the probability of $Y_{\alpha}^{\beta} > \ln(mn)$.  This motivates us to consider $\mathrm{E}[Y_{\alpha}^{\beta}(q)] > \ln(mn)/(1-\mu)$.  We also want the probability bound $1 - e^{-\mu^2\mathrm{E}[Y_{\alpha}^{\beta}]/2}$ of $Y_{\alpha}^{\beta} > \ln(mn)$ to be at least $1 - m^{-5}n^{-5}$.  This allows us to apply the union bound over at most $mn\ln(mn)(mn\ln(mn) - 1)$ choices of $\alpha$ and $\beta$ to obtain a probability bound of at least $1 - \ln^2(mn)/(m^3n^3)$.  Therefore, as we consider $\mathrm{E}[Y^{\beta}_{\alpha}] > \ln(mn)/(1-\mu)$, we want $1 - e^{-\mu^2\ln(mn)/(2(1-\mu))} = 1 - (mn)^{-\mu^2/(2(1-\mu))} = 1 - m^{-5}n^{-5}$.  Equivalently, we require $\mu^2/(2(1-\mu)) = 5$ which is satisfied by setting $\mu = \sqrt{35}-5$.  We conclude that:
	\begin{quote}
		It holds with probability at least $1-\ln^2 (mn)/(m^3n^3)$ that for any $\alpha,\beta \in [1,mn\ln(mn)]$ such that $s_\alpha < s_\beta$, if $\mathrm{E}[Y_{\alpha}^{\beta}] > \frac{1}{6-\sqrt{35}}\ln(mn)$, then $Y_{\alpha}^{\beta} > \ln(mn)$.
	\end{quote}
	
	For every $r \in [0,mn+1]$, let $s_{\alpha_r}$ denote $v_r$, where $s_{\alpha_0} = -\infty$ and $s_{\alpha_{mn+1}} = \infty$.  Fix a particular $r \in [0,mn]$.  By construction, there are at most $\ln(mn)$ numbers among $s_1, \cdots, s_{mn\ln(mn)}$ that fall in $[v_r,v_{r+1})$, which guarantees the event of $Y_{\alpha_r}^{\alpha_{r+1}} \leq \ln(mn)$.  Our previous conclusion implies that:
	\begin{quote}
	It holds with probability at least $1-\ln^2 (mn)/(m^3n^3)$ that $\mathrm{E}[Y_{\alpha_r}^{\alpha_{r+1}}] \leq \frac{1}{6-\sqrt{35}}\ln(mn)$.  
	\end{quote}
	
	The random process that generates the input is independent of the training phase.  In the training phase, for each $i \in [1,n]$, we sample $m\ln(mn)$ $x_i$'s from $m\ln(mn)$ input instances to form $(s_1,\ldots,s_{mn\ln(mn)})$.  Therefore, 
	\begin{equation}
	\mathrm{E}[Y_{\alpha_r}^{\alpha_{r+1}}] \geq \left(\sum_{i=1}^n m\ln(mn) \cdot \pr{X_{ir} = 1}\right) - 2 
	\label{eq:N-1}
	\end{equation}
	because ${\cal J}_{\alpha_r}^{\alpha_{r+1}}$ excludes $\alpha$ and $\beta$, but $s_\alpha$ and $s_\beta$ are allowed in $\sum_{i=1}^n m\ln(mn) \cdot \pr{X_{ir} = 1}$.  Observe that 
	\[
	\sum_{i=1}^n \pr{X_{ir} = 1} = 
	\sum_{i=1}^n \sum_{q=1}^\kappa \pr{X_{ir}=1 \wedge I \sim {\cal D}_q} =  \sum_{q=1}^\kappa \sum_{i=1}^n \pr{X_{ir}=1 \wedge I \sim {\cal D}_q}.
	\]
	Rerranging terms in \eqref{eq:N-1} and applying the inequality $\mathrm{E}[Y_{\alpha_r}^{\alpha_{r+1}}] \leq \frac{1}{6-\sqrt{35}}\ln(mn)$
	give
	\[
	\sum_{q=1}^\kappa \sum_{i=1}^n \pr{X_{ir}=1 \wedge I \sim {\cal D}_q} \,\, \leq \,\,
	\frac{\mathrm{E}[Y_{\alpha_r}^{\alpha_{r+1}}]}{m\ln(mn)} + \frac{2}{m\ln(mn)} \,\, = \,\, O(1/m).
	\]
	Apply the union bound over $r \in [0,mn]$.  The probability bound is thus at least $1 - (mn+1)\ln^2(mn)/(m^3n^3) \geq 1-1/(mn)$.
\end{proof}  

Recall that $N_r$ is the subset of points that fall into $[v_r,v_{r+1})$ in the operation phase when sorting an input instance.  We bound the expected total time $\mathrm{E}\left[\sum_{r=0}^{mn} |N_r| \log |N_r|\right]$ to sort the $N_r$'s.

\begin{lemma}
	\label{lem:sortN}
	It holds with probability at least $1 - 1/(mn)$ that $\mathrm{E}\left[\sum_{r=0}^{mn} |N_r| \log |N_r|\right] = O(n)$.
\end{lemma}
\begin{proof}
	\begin{eqnarray*}
		\mathrm{E}\left[\sum_{r=0}^{mn} |N_r|\log |N_r| \right] & \leq &
		\mathrm{E}\left[\sum_{r=0}^{mn} |N_r|^2 \right] \\
		& = & \mathrm{E}\left[\sum_{r=0}^{mn} \left(\sum_{i=1}^n X_{ir} \right)\left(\sum_{j=1}^n X_{jr}\right) \right] \\
		& = & \sum_{i=1}^n\sum_{j=1}^n \sum_{r=0}^{mn} \mathrm{E}\left[X_{ir}X_{jr} \right].
	\end{eqnarray*}
	Both $X_{ir}$ and $X_{jr}$ are random indicator variables.  If $X_{ir} = 1$ and $X_{jr} = 1$, then $X_{ir}X_{jr} = 1$; otherwise, $X_{ir}X_{jr} = 0$.  Therefore,
	\begin{eqnarray*}
		\sum_{i=1}^n\sum_{j=1}^n \sum_{r=0}^{mn}\mathrm{E}\left[X_{ir}X_{jr}\right] & = &
		\sum_{i=1}^n \sum_{j=1}^n \sum_{r=0}^{mn} \pr{X_{ir} = 1 \wedge X_{jr} = 1} \\
		& = &
		\sum_{i \not= j} \sum_{r=0}^{mn} \pr{X_{ir} = 1 \wedge X_{jr} = 1} + \sum_{i=1}^n \sum_{r=0}^{mn} \pr{X_{ir} = 1}.
	\end{eqnarray*}
	Since $x_i$ must fall into one of the $mn+1$ intervals, $\sum_{r=0}^{mn} \pr{X_{ir} = 1} = 1$, which gives 
	\[
	\sum_{i=1}^n \sum_{r=0}^{mn} \pr{X_{ir} = 1} = n. 
	\]
	Let $I$ denote an input instance.  Conditioned on $i \not= j$ and $I \sim {\cal D}_q$ for some $q \in [1,\kappa]$, $X_{ir} = 1$ and $X_{jr} = 1$ are two independent events, and so $\pr{X_{ir} = 1 \wedge X_{jr} = 1 | I \sim {\cal D}_q } = \pr{X_{ir}=1 | I \sim {\cal D}_q } \cdot \pr{X_{jr} = 1 |I \sim {\cal D}_q }$.  Therefore, 
	\begin{eqnarray*}
		& & \sum_{i\not= j} \sum_{r=0}^{mn} \pr{X_{ir} = 1 \wedge X_{jr} = 1} \\
		& = & 
		\sum_{i\not= j} \sum_{r=0}^{mn} \sum_{q=1}^{\kappa} \pr{X_{ir} = 1 \wedge X_{jr} = 1 | I \sim {\cal D}_q } \cdot \pr{I \sim {\cal D}_q} \\
		& = & 
		\sum_{i \not= j} \sum_{r=0}^{mn} \sum_{q=1}^{\kappa} \pr{X_{ir} = 1 | I \sim {\cal D}_q } \cdot \pr{X_{jr} = 1 | I \sim {\cal D}_q} \cdot \pr{ I \sim {\cal D}_q}.
	\end{eqnarray*}
	We expand the outermost summation over all $i \in [1,n]$ and $j \in [1,n]$.  Also, we replace $\pr{X_{jr} = 1 | I \sim {\cal D}_q} \cdot \pr{I \sim {\cal D}_q }$ by $\pr{X_{jr} = 1 \wedge I \sim {\cal D}_q }$.  Then,
	\begin{eqnarray*}
		& & \sum_{i\not= j} \sum_{r=0}^{mn} \pr{X_{ir} = 1 \wedge X_{jr} = 1} \\
		& \leq &
		\sum_{i=1}^n \sum_{j=1}^n \sum_{r=0}^{mn} \sum_{q=1}^{\kappa} \pr{X_{ir} = 1 | I \sim {\cal D}_q } \cdot \pr{X_{jr} = 1 \wedge I \sim {\cal D}_q } \\
		& = & 
		\sum_{q=1}^{\kappa} \left(\sum_{i=1}^n \sum_{r=0}^{mn} \pr{X_{ir}=1 | I \sim {\cal D}_q } \cdot \left(\sum_{j=1}^n \pr{X_{jr}=1 \wedge I \sim {\cal D}_q }\right)\right).
	\end{eqnarray*}
	By Lemma~\ref{lem:N}, it holds with probability at least $1-1/(mn)$ that for every $q \in [1,\kappa]$ and every $r \in [0,mn]$, the quantity $\sum_{j=1}^n \pr{X_{jr}=1 \wedge I \sim {\cal D}_q }$ is $O(1/m)$.  Therefore,
	\[
	\sum_{i\not= j} \sum_{r=0}^{mn} \pr{X_{ir} = 1 \wedge X_{jr} = 1}
	= O\left(\frac{1}{m}\sum_{q=1}^{\kappa} \left(\sum_{i=1}^n \sum_{r=0}^{mn} \pr{X_{ir}=1 | I \sim {\cal D}_q } \right)\right).
	\]
	Conditioned on a product distribution, $x_i$ must fall into one of the $mn+1$ intervals, and so $\sum_{r=0}^{mn} \pr{X_{ir}=1 |I \sim {\cal D}_q} = 1$, implying that $\sum_{i=1}^n \sum_{r=0}^{mn} \pr{X_{ir}=1 |I \sim {\cal D}_q} = n$.
	We conclude that
	\[
	\sum_{i\not= j} \sum_{r=0}^{mn} \pr{X_{ir} = 1 \wedge X_{jr} = 1}
	= 
	O\left(\frac{1}{m}\sum_{q=1}^{\kappa} n \right) = 
	O(n).
	\]
	This completes the proof.
\end{proof}

By Lemmas~\ref{lem:oper} and~\ref{lem:sortN}, sorting an input instance takes $O\left(n\log\log m + \frac{1}{\eps}\sum_{i=1}^n H(P_i)\right)$ expected time with probability at least $1-1/(mn)$, where $H(P_i)$ is the entropy of the random variable $P_i$ indicating the predecessor of $x_i$ in the $V$-list.  We bound $\sum_{i=1}^n H(P_i)$ in the following.

\begin{lemma}
	\label{lem:Hi}
	$\sum_{i=1}^n H(P_i) = O(n\log m + H_\pi)$.
\end{lemma}
\begin{proof}
	Let $Q$ be a random variable with value in the range $[1,\kappa]$ that indicates the specific product distribution from which the input instance is drawn.  Let $H(Q)$ be the entropy of $Q$.  
	
	By the chain rule for conditional entropy~\cite[Proposition~2.23]{ray}, we get 
	\[
	H(P_i) \leq H(P_i) + H(Q|P_i) = H(P_i,Q) = H(Q) + H(P_i|Q).
	\]
	The entropy of $Q$ is at most the logarithm of the domain size $\kappa$ of $Q$~\cite[Theorem~2.43]{ray}.  So $H(Q) \leq \log_2 \kappa$.  It follows that $\sum_{i=1}^n H(P_i) \leq n\log_2 \kappa + \sum_{i=1}^n H(P_i | Q)$.
	
	Note that $P_1|Q, P_2|Q, \ldots, P_n|Q$ are independent from each other because a product distribution is implied by the conditioning on $Q$.  It follows that $\sum_{i=1}^n H(P_i|Q) = H(P_1,P_2,\ldots,P_n|Q)$.  Conditioning does not increase entropy~\cite[Theorem~2.38]{ray}, and so $H(P_1,P_2,\ldots,P_n|Q) \leq H(P_1,P_2,\ldots,P_n)$.  Given the sorted order of the input instance $I$, we can figure out the values of $P_1, P_2,\ldots, P_n$ in $O(n\log m)$ time by merging the sorted order of $I$ with the $V$-list as follows.  As in the operation phase, we group the $mn+1$ intervals induced by the $V$-list into $n$ buckets, each containing $m$ intervals except the last bucket which contains $m+1$ intervals.  In $O(n)$ time, we can merge the sorted order of $I$ with the ordered list of $n$ buckets.  For each number $x_i \in I$ that lies in a bucket $B$, by comparing $x_i$ with middle $v_r$ value in $B$, we decide whether $x_i$ lies in the first $m/2$ intervals in $B$ or the other intervals in $B$.  Recursively, we can place $x_i$ in an interval in $O(\log m)$ time, which gives $P_i$.  The total time needed for all $n$ input numbers is $O(n\log m)$.  Then, Lemma~\ref{lem:relate} implies that $H(P_1,P_2,\ldots,P_n) = O(n\log m + H_\pi)$.

	Hence, $\sum_{i=1}^n H(P_i) = O(n\log m + n\log \kappa + H_\pi) = O(n\log m + H_\pi)$.
\end{proof}

%A simple scan of the list $V$ to concatenate the non-empty $N_r$'s gives a limiting complexity of $O(mn + \sum_{i=1}^n t_i)$.  When $m = o(\log\log n)$, $mn$ is smaller than $n\log\log(mn)$.
%which makes the simple scan of the list $V$ a better alternative than using the van Emde Boas tree.

The limiting complexity of $O\left((n\log m)/\eps + H_\pi/\eps \right)$ as stated in Theorem~\ref{thm:2} follows Lemmas~\ref{lem:oper}, \ref{lem:sortN}, and~\ref{lem:Hi}.  The $O(mn + m^\eps n^{1+\eps})$ space needed by the operation phase follows from Lemma~\ref{lem:oper}.  In the training phase, the space usage, processing time, and the number of input instances needed follow from Lemma~\ref{lem:train-3}.  The success probability of $1 - 1/(mn)$ follows from Lemma~\ref{lem:sortN}.  This completes the proof of Theorem~\ref{thm:2}.  In the interesting special case of $m = O(1)$, the limiting complexity is $O(n/\eps + H_\pi/\eps)$ which is optimal.

\section{Conclusion}

There are several possible directions for future research.  One is to extend the hidden classification to allow the $x_i$'s in the same class $S_k$ to be more arbitrary functions in the random parameter $z_k$.  Linear functions in $z_k$ have the nice property that any $x_i$ and $x_j$ in the same class are linearly related.  This helps us to learn the hidden classes.  
Another direction is to improve the performance in the case of a hidden mixture of product distributions.  It would also be interesting to design self-improving algorithms for other problems and possibly other input settings as well.

\section*{Acknowledgment}

We thank the anonymous reviewers for their valuable comments, suggesting a cleaner proof of Lemma~\ref{lem:Hi}, and alerting us to mistakes that we subsequently corrected.  
%In making the revision, we discovered an improvement of the limiting complexity in Theorem~\ref{thm:2} for $m = o(\log n)$ from the conference version of this paper~\cite{cheng}.

\end{document}